\documentclass{article}
\usepackage[english]{babel}
\usepackage{xcolor}
\usepackage{amsmath}
\usepackage{amsfonts}
\usepackage{amsthm}
\usepackage{amssymb}
\usepackage{calrsfs}
\usepackage{graphicx}
\usepackage{bussproofs}
\usepackage{proof}
\usepackage{latexsym}
\usepackage[ND,SEQ]{prftree}
\usepackage{multicol, caption}
\usepackage{setspace}
\usepackage[nottoc,numbib]{tocbibind}
\usepackage{pdflscape}
\usepackage{stmaryrd}
\usepackage{wasysym}
\usepackage[titletoc,title]{appendix}
\newtheoremstyle{break}
  {\topsep}{\topsep}%
  {\itshape}{}%
  {\bfseries}{}%
  {\newline}{}%
\theoremstyle{break}
\newtheorem{theorem}{Theorem}[section]

\newtheorem{definition}{Definition}
\newtheorem{lemma}{Lemma}

\usepackage{parskip}

\begin{document}
\title{Proof-Theoretic Functional Completeness for the Connexive Logic \texttt{C}}

\author{Sara Ayhan \& Hrafn Oddsson\\ Ruhr University Bochum}
\date{}
\maketitle

\begin{abstract}
We show the functional completeness for the connectives of the non-trivial negation inconsistent logic \texttt{C} by using a well-established method implementing purely proof-theoretic notions only. Firstly, given that \texttt{C} contains a strong negation, expressing a notion of direct refutation, the proof needs to be applied in a bilateralist way in that not only higher-order rule schemata for proofs but also for refutations need to be considered. Secondly, given that \texttt{C} is a connexive logic we need to take a connexive understanding of inference as a basis, leading to a different conception of (higher-order) refutation than is usually employed. 
\end{abstract}

\section{Introduction}
From the viewpoint of what would nowadays be called `(bilateralist) proof-theoretic semantics' Franz von Kutschera provided two very relevant and closely related papers \cite{Kutschera1968,Kutschera1969}, in which he proves functional completeness of, respectively, intuitionistic logic and what he calls `direct logic', nowadays better known as Nelson's constructive logic with strong negation, \texttt{N4}.
Proof-theoretic semantics is the view that it is the rules of logical connectives governing their use in a proof system that determine their meaning.\footnote{Not unwarrantedly, von Kutschera called this `Gentzen semantics', since this view's origins are usually attributed to Gentzen's famous paper \cite{Gentzen1934}. For more details on proof-theoretic semantics, see \cite{s-h} and \cite{Francez2015}.}
Von Kutschera commits to this view in these papers in that he assumes rule schemata to define connectives and in conducting the proofs for functional completeness purely proof-theoretically, i.e., without reference to model-theoretic notions.
The procedure is in both papers essentially the same.
In a nutshell, he considers generalized derivability schemata for unspecified logical operators and shows that the connectives of intuitionistic propositional logic and \texttt{N4}, respectively, are functionally complete when defined according to these rule schemata.

The aim of the present paper is to transfer this approach to show the functional completeness for the connectives of the connexive logic \texttt{C}.
As such, our paper is in a tradition of several others using von Kutschera's approach in one way or another; to name just a few: 
Schroeder-Heister, combining von Kutschera's higher level derivability with Prawitz's natural deduction framework, shows completeness for intuitionistic propositional logic in such a system in
\cite{PSH1984b} and for an extension to intuitionistic quantifier logic in \cite{PSH1984a}. 
Wansing uses von Kutschera's approach to show the functional completeness of operators for substructural subsystems of intuitionistic propositional logic \cite{Wansing1993a} and of \texttt{N4} \cite[Ch. 7]{Wansing1993b} and also for several normal modal and tense logics with respect to a representation of the generalized rules in display sequent calculi \cite{Wansing1996}.
Braüner \cite{Brauner2005} does the same for hybridized versions of certain modal logics.

It is especially von Kutschera's later paper \cite{Kutschera1969} that is interesting to consider from a \emph{bilateralist} proof-theoretic semantics point of view and also for the present purpose.\footnote{See \cite{TranslationKutschera} for a translation of this paper as well as \cite{Ayhancomment} for a comment on its relevance and its connection to \cite{Kutschera1968}.}
In this paper von Kutschera is concerned with giving a concept of refutation that is -- other than in intuitionistic logic -- not derivative on a concept of proof but on a par with it.
This is a very bilateralist stance, since bilateralism is the view that in proof-theoretic semantics it is not only the rules governing proofs of complex formulas but also the rules governing their refutations that need to be considered as meaning-determining and, importantly, that these concepts should be considered both as primitive, not one reducible to the other.\footnote{More often in bilateralism the terms `assertion' and `denial' are used as the relevant concepts but like von Kutschera we will stick to the terms `proof' and `refutation'.}
Accordingly, in \cite{Kutschera1969}, in order to define arbitrary $n$-ary connectives, he gives generalized rule schemata not only to introduce \emph{proofs} but also to introduce \emph{refutations}. 

We will start with a brief introduction to our system, considering the idea of connexive logics in general, the fact that \texttt{C} is a contradictory logic and the relevant differences to \texttt{N4} and to von Kutschera's approach (Section 2). 
To provide a basis for our proof of functional completeness, we will firstly consider generalized derivability schemata for unspecified logical operators as von Kutschera's schemata but with a different background notion of higher-level refutability reflecting the connexive reading of consequence (Section 3-4). 
Subsequently, we trace von Kutschera's proof of functional completeness for the \texttt{N4} connectives, yielding the result that the connectives of \texttt{C} are functionally complete for these generalized schemata (Section 5-6).

\section{The motivation: A contradictory logic that is functionally complete}

There has been a recent increase in attention with respect to studying connexive logics, resulting in a number of papers published on the subject.\footnote{See, e.g., the contributions in the special issue of \emph{Studia Logica} on \emph{Frontiers of Connexive Logic} \cite{OmoriWansingSI}.}
The general idea of connexivity that should be implemented by these systems, is that no formula provably implies or is implied by its own negation. Formally this is usually realized by validating so-called \emph{Aristotle's} (AT) and \emph{Boethius' Theses} (BT), while not validating symmetry of implication.
Hence, the following is the most common, agreed-upon definition:

\begin{definition}[Connexive logics]
 A logic is connexive iff 
 
 1.) the following formulas are provable in the system:
\begin{itemize}
		\item AT: ${\sim} ({\sim} A \rightarrow A)$,
		\item AT': ${\sim} (A \rightarrow {\sim} A)$,
		\item BT: $(A \rightarrow B) \rightarrow {\sim}(A \rightarrow {\sim} B)$,
		\item BT': $(A \rightarrow {\sim} B) \rightarrow {\sim}(A \rightarrow B)$, and
\end{itemize}

2.) the system satisfies non-symmetry of implication, i.e., $(A\rightarrow B) \rightarrow (B \rightarrow A)$ is not provable.
\end{definition}
We will concentrate here on the system \texttt{C}, introduced in \cite{Wansing2005}, which is obtained by changing the falsification conditions of implication in Nelson's four-valued constructive logic with strong negation, \texttt{N4} \cite{AN1984, Nelson1949}.\footnote{To give more context and references for the system \texttt{C}: See, e.g., \cite{Wansing2024} and \cite{WansingWeber} for motivations for \texttt{C}, \cite{Olkhovikov2023}, where a direct completeness proof for 1st-order \texttt{C} is given, or \cite{FazioOdintsov2024} for an algebraic investigation of the system. For a sequent calculus for \texttt{C}, see Appendix.}

The following is a natural deduction calculus for the propositional fragment of \texttt{C}, which is adapted from a bilateral version given in \cite[p. 419f.]{Wansing2016}. 
Although in the paper introducing this system \cite{Wansing2005}, the aim is to present a connexive \emph{modal} logic, we will not be concerned with modal operators here and rather concentrate on the way this system is constructed from \texttt{N4}.

\vspace{0.5cm}

{\Large\textbf{\texttt{NC}}} 
\vspace{0.2cm}

\textbf{Rules for ${\sim}$:}
\vspace{0.2cm}

\quad   
\infer[\scriptstyle{\sim} {\sim} I]{{\sim} {\sim} A}{\infer*{A}{\Gamma}}
\quad 
\infer[\scriptstyle{\sim}{\sim} E]{A}{
	\infer*{{\sim} {\sim} A}{\Gamma}}

\vspace{0.2cm}

\textbf{Rules for $\wedge$:}

\vspace{0.2cm}
\quad  
\infer[\scriptstyle\wedge I]{A \wedge B}{\;\infer*{A}{\Gamma} \quad \quad \infer*{B}{\Delta}}
\quad   
\infer[\scriptstyle\wedge E_{1}]{A}{\infer*{A \wedge B}{\Gamma}}
\quad   
\infer[\scriptstyle\wedge E_{2}]{B}{\infer*{A \wedge B}{\Gamma}}

\vspace{0.2cm}
\quad   
\infer[\scriptstyle{\sim}\wedge I_{1}]{{\sim}(A \wedge B)}{\infer*{{\sim} A}{\Gamma}}
\quad \quad
\infer[\scriptstyle{\sim}\wedge I_{2}]{{\sim}(A \wedge B)}{\infer*{{\sim} B}{\Gamma}}
\quad
\hspace{-0.5cm}  
\infer[\scriptstyle{\sim}\wedge E]{C}{\;\;\;\;\;\;\;\;\;\infer*{{\sim}(A \wedge B)}{\Gamma} \quad \infer*{C}{[{\sim} A], \Delta} \quad \infer*{C}{[{\sim} B], \Theta }}
\vspace{0.2cm}

\textbf{Rules for $\vee$:}

\vspace{0.2cm}
\quad  
\infer[\scriptstyle\vee I_{1}]{A \vee B}{\infer*{A}{\Gamma}}
\quad  \quad
\infer[\scriptstyle\vee I_{2}]{A \vee B}{\infer*{B}{\Gamma}}
\quad 
\hspace{-0.5cm} 
\infer[\scriptstyle\vee E]{C}{\;\;\;\;\;\;\;\;\;\infer*{A \vee B}{\Gamma} \quad \infer*{C}{ [A], \Delta} \quad \infer*{C}{[B], \Theta}}

\vspace{0.2cm}

\quad  
\infer[\scriptstyle{\sim}\vee I]{{\sim}(A \vee B)}{\;\infer*{{\sim} A}{\Gamma}\quad \quad\infer*{{\sim} B}{\Delta} }
\quad  \quad \quad
\infer[\scriptstyle{\sim}\vee E_{1}]{{\sim} A}{\infer*{{\sim}(A \vee B)}{\Gamma}}
\quad  \quad
\infer[\scriptstyle{\sim}\vee E_{2}]{{\sim} B}{\infer*{{\sim}(A \vee B)}{\Gamma}}
\vspace{0.2cm}

\textbf{Rules for $\rightarrow$:}

\vspace{0.2cm}

\quad 
\infer[\scriptstyle\rightarrow I]{A \rightarrow B}{
	\infer*{B}{[A],\Gamma}}
\quad \hspace{-0.35cm} 
\infer[\scriptstyle\rightarrow E]{B}{\infer*{A}{\Gamma} \ \quad \infer*{A \rightarrow B}{\Delta}}
\quad \hspace{-0.35cm}
\infer[\scriptstyle{\sim}\rightarrow I]{{\sim}(A \rightarrow B)}{
	\infer*{{\sim} B}{[A],\Gamma}}
\quad \hspace{-0.35cm} 
\infer[\scriptstyle{\sim}\rightarrow E]{{\sim} B}{\infer*{A}{\Gamma} \ \quad \infer*{{\sim}(A \rightarrow B)}{\Delta}}

 The difference between a natural deduction system  for \texttt{N4}, as can be found in \cite[p. 97]{Prawitz1965}, for example, and \texttt{NC} lies in the rules for negated implication.
 We obtain a connexive system by ${\sim\rightarrow} I$ and ${\sim\rightarrow} E$, as given above, replacing these rules in \texttt{N4}:
 
 \quad 
\infer[\scriptstyle\sim\rightarrow I]{\sim(A \rightarrow B)}{{A}\quad\quad {\sim B}}
\quad \quad
\infer[\scriptstyle\sim\rightarrow E_{1}]{A}{\sim(A \rightarrow B)}
\quad \quad
\infer[\scriptstyle\sim\rightarrow E_{2}]{\sim B}{\sim(A \rightarrow B)}

Von Kutschera takes \texttt{N4} to express naturally a concept of direct refutation, represented by the connective ${\sim}$, which we call `strong negation'.
However, with \texttt{C} being connexive, the interpretation of what it means to refute an implication $A \rightarrow B$ is very different from \texttt{N4}, where this means to have a proof of $A$ and ${\sim} B$.
With a connexive implication refuting an implication $A \rightarrow B$ means to derive a refutation of $B$ from the assumption $A$.\footnote{To be more precise, actually this notion captures what is sometimes called a ``hyperconnexive'' understanding of implication, which is not only reflected in the validity of Boethius' thesis $(A\rightarrow B) \rightarrow {\sim}(A\rightarrow {\sim} B)$ as a characteristic principle of connexive logic, but also in the validity of its converse. For remarks and an overview on the differing terminology in connexive logic, see \cite{WansingOmori2024}.}
Thus, this needs to be reflected in a generalized conception of `higher-order' refutability, too.
Another peculiarity arising from this is that \texttt{C} is a \emph{contradictory logic}.
We use this term for logics going beyond paraconsistency by being not only not explosive but also containing contradictory formulas (here: formulas of the form $A$ and ${\sim} A$) in their set of theorems while still being non-trivial.\footnote{See, e.g., \cite{Wansingforthcoming}, where it is argued that it is theoretically rational to work with non-trivial negation inconsistent logics and a list of such logics is presented, or \cite{NikiWansing2023}, for examples and a classification of provable contradictions in \texttt{C} and an extension of \texttt{C}.} To give a simple example of such a provable contradiction in \texttt{C}:
\vspace{0.2cm}

\begin{center}
    
\quad
\infer[\scriptstyle \rightarrow I]{(A \wedge {\sim} A) \rightarrow A}{\infer[\scriptstyle \wedge E_1]{A}{{[A \wedge {\sim} A]}}}
\quad \quad
\infer[\scriptstyle {\sim}\rightarrow I]{\sim((A \wedge {\sim} A) \rightarrow A)}{\infer[\scriptstyle \wedge E_2]{{\sim} A}{{[A \wedge {\sim} A]}}}
\end{center}

The proof of functional completeness for \texttt{C} can essentially proceed in the fashion of von Kutschera's, so we do not think it is necessary to provide it in full length here but we will restrict the paper to a sketch of the proof while making the differences explicit. 

There is one important conceptual point worth mentioning, though.
When discussing under what conditions we can give rules for refutations of arbitrary connectives, von Kutschera says that ``[o]bviously, these rules [for refutation] cannot be devised independently of [the rules for proofs] if the consistency of the Gentzen calculi is to be maintained when these rules are added'' \cite[p. 108, (translated by the authors)]{Kutschera1969}.
What he fails to mention in this context is that our conception of how precisely to refute inferences plays a crucial role here.
We also give the refutation rules dependently on the proof rules in the same manner as it is done in \cite{Kutschera1969} and we do get the completeness result for the \emph{inconsistent} logic \texttt{C}.
So, in light of our adjustments of von Kutschera's approach for \texttt{C} (Section 4), the dependence of the refutation rules on the proof rules seems  to be an independent matter, separate from any considerations about consistency.
It is rather the underlying notion of refutation that is important here and with respect to \emph{that} our approaches are obviously essentially different.
While von Kutschera (without any further explanation or justification) assumes a notion of refutation that captures the, in many logics, usual way of interpreting a negated implication, i.e., the refutability of a derivation going from $A$ to $B$ is defined by the provability of $A$ and the refutability of $B$, our notion will capture the connexive understanding of implication.
Thus, the refutability of a derivation going from $A$ to $B$ will be defined by the refutability of $B$ given the assumed provability of $A$.
So, it is rather von Kutschera's underlying assumption that derivability should basically behave as implication does in \texttt{N4} that is critical for consistency, not the factor of dependent definitions between proofs and refutations.

\section{The language}
We suppose a language $L$ consisting of propositional constants and connectives for which there are formulas specified.
We will inductively define an $R$-expression over $L$ as follows:

1. Every formula over $L$ is an $R$-expression.

2. If $S$ is an $R$-expression over $L$ which does not have the form $-T$, then so is $-S$ (which can be read as `refutation of $S$').

3. If $S_1,..., S_n, T$ are $R$-expressions over $L$, then so is ($S_1,..., S_n \Rightarrow T$).

4. Nothing else is an $R$-expression over $L$.

\noindent Outer brackets of $R$-expressions may be omitted.
We will use $A, B, C,...$ for formulas and $S, T, R,...$ for $R$-expressions.
For (possibly empty) series of $R$-expressions separated by commas $\Gamma, \Delta, \Theta, ...$ are used.

\noindent The \emph{R-degree} of $R$-expressions is now defined as follows.
A formula of $L$ has the $R$-degree 0.
If $S$ has $R$-degree $n$, then so does $-S$.
If $n$ is the maximum of the $R$-degrees of $\Delta, S$, then $n+1$ is the $R$-degree of $\Delta \Rightarrow S$.
$R$-expressions of $R$-degree 1 will be called \emph{sequents}. 

Furthermore, we will define the following: $S$ is \emph{R-subformula} of $S$ and all $R$-subformulas of $S$, resp. $\Delta$, $S$, are also $R$-subformulas of $-S$, resp. $\Delta \Rightarrow S$.
$R$-subformulas of $R$-degree 0 of an $R$-formula $S$ will be called the \emph{formula components} of $S$.

\section{Higher-Order Derivability}

Now we want to define the notion of derivability for $R$-expressions of arbitrary $R$-degree.
Therefore, we define a calculus in which not only sequents but $R$-expressions of arbitrary degree are allowed. 
The following structural rules hold:

\begin{center}
\vspace{0.4cm}
\quad \infer[RF]{S\Rightarrow S}{}\hspace{7em}\infer[WL]{\Delta, T \Rightarrow S}{\Delta \Rightarrow S} 

\vspace{0.2cm}
\quad\infer[PL]{\Delta, T, S, \Gamma \Rightarrow U}{\Delta, S, T, \Gamma \Rightarrow U}\hspace{6em} \infer[CL]{\Delta, S \Rightarrow T}{\Delta, S, S \Rightarrow T }

\vspace{0.2cm}
\quad \infer[Cut]{\Delta \Rightarrow T}{\Delta \Rightarrow S \qquad \Delta, S \Rightarrow T}
\end{center}

\noindent 

We add two rules, which tell us how to introduce an $R$-expression of arbitrary degree on the right and on the left side of $\Rightarrow$:

\begin{center}
\vspace{0.3cm}
\quad \infer[RI^+]{\Gamma \Rightarrow (\Delta \Rightarrow S)}{\Gamma, \Delta \Rightarrow S}\quad\quad \infer[LI^+]{ \Gamma, (\Delta \Rightarrow S) \Rightarrow T}{\Gamma \Rightarrow \Delta \quad\quad \Gamma, S \Rightarrow T}
\vspace{0.2cm}
\end{center}

\noindent The expression $\Gamma \Rightarrow \Delta$ is here to be read as  $\Gamma \Rightarrow U_1;...;\Gamma \Rightarrow U_n$ where $\Delta$ is the series of $R$-expressions $U_1,...U_n$.

\noindent In line with our connexive understanding of inference, we define the refutability of $\Gamma \Rightarrow S$ by the refutability of $S$ under the assumption of the provability of $\Gamma$ and the refutability of $\Gamma \Rightarrow -S$ by the provability of $S$ under the assumption of the provability of $\Gamma$.
To formulate this as one rule, we can fix that $--S$ stands for $S$.
Then we have the following rule:

\begin{center}
\vspace{0.3cm}
\quad\infer[RI^-]{ \Delta \Rightarrow -(\Gamma \Rightarrow S)}{\Delta, \Gamma \Rightarrow -S}
\vspace{0.2cm}
\end{center}

\noindent If the reverse is demanded to hold as well, i.e.:

\begin{center}
\vspace{0.2cm}
\quad\infer{\Delta \Rightarrow -(\Gamma \Rightarrow S)}{ \Delta, \Gamma \Rightarrow -S}
\vspace{0.2cm}
\end{center}

\noindent then we also obtain a condition on how to derive something \emph{from} $R$-expressions of the form $-(\Gamma \Rightarrow S)$:

\begin{center}
\vspace{0.2cm}
\quad\infer[LI^-]{\Delta, -(\Gamma \Rightarrow S) \Rightarrow T}{\Delta \Rightarrow \Gamma \quad\quad \Delta, -S \Rightarrow T}
\vspace{0.2cm}
\end{center}

We will write $\Delta\vdash S$ to mean that there is a derivation from the $R$-expressions of $\Delta$ to $S$, and we write $\Delta\vdash \Gamma$ to mean that $\Delta\vdash S$ for each $S$ from $\Gamma$.

Now that we have our deductive framework, let us briefly outline the remainder of the paper.
First, we adopt von Kutschera's generalized introduction rules for arbitrary connectives.
We then give the introduction rules for the usual connectives for \texttt{C} and follow von Kutschera's argument step by step to obtain functional completeness for \texttt{C}, meaning that any connective expressible by the generalized schemata can be explicitly defined through the connectives of \texttt{C}.
Finally, we provide a way to translate this framework into a standard sequent calculus of \texttt{C} and obtain that a sequent is derivable in this framework iff its translation is derivable in \texttt{C}.

\section{Defining the logical operators}
Now we extend our calculus to a calculus $\texttt{SC}_\infty$, which allows us to define an $n$-ary connective $F$ by giving general schemata of rules telling us how to introduce expressions of the form $F(A_1,..., A_n)$ and of the form $-F(A_1,..., A_n)$ on the right or left side of $\Rightarrow$. 
For the formulations of the schemata for rules to introduce $F(A_1,..., A_n)$ we will follow \cite[p. 108]{Kutschera1969}:

\vspace{0.2cm}
(I)
\begin{center}
    
\quad \infer{\Delta \Rightarrow F(A_1,..., A_n)}{\Delta \Rightarrow S_{1 1}\quad ...\quad \Delta \Rightarrow S_{1 s_{1}}}\quad ... \quad \infer{\Delta \Rightarrow F(A_1,..., A_n)}{\Delta \Rightarrow S_{t 1}\quad ...\quad \Delta \Rightarrow S_{t s_{t}}}

\end{center}

\vspace{0.2cm}
(II)
\begin{center}
    \quad \infer{\Delta, F(A_1,..., A_n) \Rightarrow T}{\Delta, S_{1 1},..., S_{1 s_{1}} \Rightarrow T \quad ... \quad \Delta, S_{t 1}, ..., S_{t s_{t}} \Rightarrow T }
\end{center}

\vspace{0.2cm}

\noindent where expressions of the form $S_{i k_{i}}$ ($i=1,...,t; k_{i}=1,...,s_{i}$) are $R$-expressions whose components are formulas from $A_1,...,A_n$, while the components of $\Delta$ are unspecified.

\noindent Also following \cite[p. 108]{Kutschera1969}, the rules for $-F(A_1,..., A_n)$ will depend on I) and II): 

\vspace{0.2cm}
(III)
\begin{center}
    \quad\infer{\Delta \Rightarrow -F(A_1,..., A_n)}{\Delta \Rightarrow -S_{1 k_{11}}\quad ...\quad \Delta \Rightarrow -S_{t k_{t1}}}\quad ... \quad \infer{\Delta \Rightarrow -F(A_1,..., A_n)}{\Delta \Rightarrow -S_{1 k_{1r}}\quad ...\quad \Delta \Rightarrow -S_{t k_{tr}}}
\end{center}
\vspace{0.2cm}

\noindent where $r=s_1\times ...\times s_t$ and $k_{i l}=1, ..., r$.

\vspace{0.2cm}
(IV)
\begin{center}
\quad\infer{\Delta, -F(A_1,..., A_n) \Rightarrow T}{\Delta, -S_{1 k_{11}},..., -S_{t t_{t1}} \Rightarrow T\quad ... \quad \Delta, -S_{1 k_{1r}},..., -S_{t k_{1r}} \Rightarrow T}
\end{center}
\vspace{0.2cm}

Unlike in \cite{Kutschera1969}, though, the schemata  III) and IV) are not motivated by a desire for consistency. Rather, we make a decision to read off III) and IV) as the De Morgan duals of I) and II).

Before going any further, let us start by noting some technical results telling us when we can substitute one $R$-expression for another. Their proofs are all straightforward adaptations from \cite{Kutschera1968} and \cite{Kutschera1969}. 

\noindent We say that $S \Rightarrow^{s} T$ stands for $S \Rightarrow T$ and $-T \Rightarrow -S$, and $S \Leftrightarrow^{s} T$ for $S \Rightarrow^{s} T$ and $T \Rightarrow^{s} S$. 
If $S \Leftrightarrow^{s} T$, we call $S$ and $T$ \emph{strictly equivalent}.
Then we get the following substitution theorem for the calculus:

\begin{theorem}\label{substitution}$S \Leftrightarrow^{s} T \vdash U_S \Leftrightarrow^{s} U_T$.
\end{theorem} 

\noindent $U_S$ is here an $R$-expression containing an occurrence of $S$ as $R$-subformula and $U_T$ is obtained from $U_S$ by replacing this occurrence by $T$. 
The proof is identical to the one given in \cite[p. 8]{Kutschera1968}.

\noindent On grounds of the rules (I)-(IV), we can now also prove the following substitution theorem:

\begin{theorem}\label{substitution2}
$ A \Leftrightarrow^{s} B \vdash  C_A \Leftrightarrow^{s} C_B$.
\end{theorem}

\noindent Here, $C_A$ is a formula containing a particular occurrence of $A$ and $C_B$ the result of replacing this occurrence in $C_A$ by $B$. 

Together with Theorem \ref{substitution}, we also get this more generalized version of the substitution theorem:

\begin{theorem}\label{substitution3}
$ A \Leftrightarrow^{s} B \vdash S_A \Leftrightarrow^{s} S_B$ (with $S_A$ being an $R$-expression containing an occurrence of formula $A$).
\end{theorem}

\noindent Now the operators $\{{\sim}, \wedge, \vee, \rightarrow\}$ shall be defined according to schemata (I)-(IV) with the following rules:

\vspace{0.2cm}
\quad  \infer[{\sim} L]{\Delta, {\sim} A \Rightarrow S}{\Delta, -A \Rightarrow S } \qquad \infer[{\sim} R]{\Delta \Rightarrow {\sim} A}{\Delta \Rightarrow -A}

\vspace{0.2cm}
\quad \infer[{\sim} L_-]{\Delta, -{\sim} A \Rightarrow S}{\Delta, A \Rightarrow S} \qquad \infer[{\sim} R_-]{\Delta \Rightarrow -{\sim} A}{\Delta \Rightarrow A}

\vspace{0.2cm}
\quad \infer[\wedge L]{\Delta, A \wedge B \Rightarrow S}{\Delta, A, B \Rightarrow S} \qquad \infer[\wedge R]{\Delta \Rightarrow A \wedge B}{\Delta \Rightarrow A \qquad \Delta \Rightarrow B}

\vspace{0.2cm}
\quad \infer[\wedge L_-]{\Delta, -(A \wedge B) \Rightarrow S}{\Delta, -A \Rightarrow S \quad \Delta, -B \Rightarrow S}\quad \infer[\wedge R_-]{\Delta \Rightarrow -(A \wedge B)}{\Delta \Rightarrow -A}\quad \infer[\wedge R_-]{ \Delta \Rightarrow -(A \wedge B)}{\Delta \Rightarrow -B}

\vspace{0.2cm}
\quad \infer[\vee L]{\Delta, A \vee B \Rightarrow S}{\Delta, A \Rightarrow S \qquad \Delta, B \Rightarrow S}\qquad \infer[\vee R]{\Delta \Rightarrow A \vee B}{\Delta \Rightarrow A}\quad \infer[\vee R]{\Delta \Rightarrow A \vee B}{\Delta \Rightarrow B}

\vspace{0.2cm}
\quad \infer[\vee L_-]{\Delta, -(A \vee B) \Rightarrow S}{\Delta, -A, -B \Rightarrow S} \qquad \infer[\vee R_-]{\Delta \Rightarrow -(A \vee B)}{\Delta \Rightarrow -A \qquad \Delta \Rightarrow -B}

\vspace{0.2cm}
\quad \infer[\rightarrow L]{\Delta, (A \rightarrow B) \Rightarrow S}{\Delta, (A \Rightarrow B) \Rightarrow S} \qquad \infer[\rightarrow R]{\Delta \Rightarrow (A \rightarrow B)}{\Delta \Rightarrow (A \Rightarrow B)}

\vspace{0.2cm}
\quad \infer[\rightarrow L_-]{\Delta, -(A \rightarrow B)\Rightarrow S}{\Delta, -(A \Rightarrow B)\Rightarrow S} \qquad \infer[\rightarrow R_-]{\Delta \Rightarrow -(A \rightarrow B)}{\Delta, -(A \Rightarrow B)}
\vspace{0.2cm}

The rules for implication are equivalent to the following, which are in a more standard formulation: 

\vspace{0.2cm}
\quad \infer[\rightarrow L^*]{\Delta, A \rightarrow B \Rightarrow S}{\Delta \Rightarrow A \qquad \Delta, B \Rightarrow S} \qquad \infer[\rightarrow R^*]{\Delta \Rightarrow A \rightarrow B}{\Delta, A \Rightarrow B}

\vspace{0.2cm}
\quad \infer[\rightarrow L^*_-]{\Delta, -(A \rightarrow B) \Rightarrow S}{\Delta \Rightarrow A \qquad \Delta, -B \Rightarrow S} \qquad \infer[\rightarrow R^*_-]{\Delta \Rightarrow -(A \rightarrow B)}{\Delta, A \Rightarrow -B}

\vspace{0.2cm}
From this it follows that: $\vdash A \rightarrow B \Leftrightarrow^{s} A \Rightarrow B$ and $\vdash {\sim} A \Leftrightarrow^{s} -A$.
 
\section{Functional completeness}
The problem of functional completeness for a given logic $\mathcal{L}$ involves identifying a set $\Theta$ of logical connectives of $\mathcal{L}$ such that every logical connective of $\mathcal{L}$ can be explicitly defined through a finite number of compositions using elements of $\Theta$.
We will show that the system of operators $\{{\sim}, \wedge, \vee, \rightarrow\}$ of \texttt{C} is functionally complete.

Following \cite{Kutschera1969}, we start by noting that we can express the structural symbols of our system in terms of the connectives from $\{{\sim}, \wedge, \vee, \rightarrow\}$ of \texttt{C} in the following way.

\noindent For every $R$-expression $S$ we have a mapping onto a formula $\overline{S}$ according to the following rules:

1.) If $S$ is a formula, then $\overline{S}=S$.

2.) $\overline{-S}={\sim}\overline{S}$.

3.) If $\Delta$ is the series of $R$-expressions $S_1,..., S_n$, then $\overline{\Delta}=\overline{S_1},..., \overline{S_n}$ and $\overline{\overline{\Delta}}=\overline{S_1}\wedge...\wedge \overline{S_n}$.

4.) $\overline{\Delta \Rightarrow S}=\overline{\overline{\Delta}}\rightarrow \overline{S}$.

5.) $\overline{\Rightarrow S}=\overline{S}$.

Then the following holds:

\begin{theorem}\label{substitution4} 
$\vdash\overline{S} \Leftrightarrow^{s} S$.
\end{theorem}

\noindent The proof is a simple induction on the complexity of $S$. It is identical to \cite[p. 109f.]{Kutschera1969}. By a standard argument (see \cite[p. 110]{Kutschera1969}), we get the following:


\begin{theorem}\label{substitution5}
 $\vdash F(A_1,...,A_n) \Leftrightarrow^{s} (\overline{S_{1 1}} \wedge... \wedge \overline{S_{1 s_{1}}}) \vee...\vee(\overline{S_{t 1}} \wedge...\wedge \overline{S_{t s_{t}}}).$
\end{theorem}

By Theorem \ref{substitution3}, this means that  $ F(A_1,...,A_n) $ can be replaced in all contexts by the formula $(\overline{S_{1 1}} \wedge... \wedge \overline{S_{1 s_{1}}}) \vee...\vee(\overline{S_{t 1}} \wedge...\wedge \overline{S_{t s_{t}}})$, i.e., is defined by that formula, which is constituted just by the formulas $A_1,...,A_n$ and the connectives $\wedge, \vee, \rightarrow, {\sim}$.
The system with these operators is thus functionally complete.

All that remains to show is that our system actually corresponds to \texttt{C}.
We will show that the sequents\footnote{Recall that `sequents' in our higher-order sequent calculus refer to $R$-expressions of degree 1.} containing only connectives from $\{\wedge, \vee, \rightarrow, {\sim}\} $ derivable in $\texttt{SC}_\infty$ correspond exactly to the ones derivable in \texttt{G3C}, the sequent calculus for \texttt{C} spelled out in the appendix.

First, we will need a bit more notation. For every formula $A$, we define the formula $A^*$ as follows:

1.) $p^*=p$, for each propositional variable $p$.

2.) $(F(A_1,...,A_n))^*$ is obtained from $(\overline{S_{1 1}} \wedge... \wedge \overline{S_{1 s_{1}}}) \vee...\vee(\overline{S_{t 1}} \wedge...\wedge \overline{S_{t s_{t}}})$ by simultaneously replacing each $A_i$ with $A_i^*$.

Finally, for any $R$-expression $S$, we let $S^*=(\overline{S})^*$ and if $\Delta$ is the series of $R$-expressions $S_1,..., S_n$, then $\Delta^*=S_1^*,..., S_n^*.$

\begin{lemma}\label{Lemma}
    An $R$-expression $\Delta\Rightarrow S$ is derivable in $\texttt{SC}_\infty$ iff the sequent $\Delta^*\Rightarrow S^*$ is derivable in \texttt{G3C}.
\end{lemma}
\begin{proof}
    Clearly, every rule of \texttt{G3C} is derivable in $\texttt{SC}_\infty$. So, if $\vdash_{\texttt{G3C}}\Delta^*\Rightarrow S^*$, then $\vdash_{\texttt{SC}_\infty} \Delta^*\Rightarrow S^*$. It follows from Theorem \ref{substitution4} that $\vdash_{\texttt{SC}_\infty} \Delta\Rightarrow S$.

    Next, we show by the length of the derivation that $\vdash_{\texttt{SC}_\infty} \Delta\Rightarrow S$ implies $\vdash_{\texttt{G3C}}\Delta^*\Rightarrow S^*:$ If $\Delta\Rightarrow S$ is an initial sequent, i.e., of the form $T\Rightarrow T$, then $\Delta^* \Rightarrow S^*$ is of the form $T^*\Rightarrow T^*$, 
    which is easily derivable in \texttt{G3C}.
    The induction step for the remaining structural rules ($WL$, $PL$, $CL$, $Cut$) is trivial. 
    
    If the last step of the derivation is an application of $RI^+$, then $\Delta\Rightarrow S$ is of the form  $ \Gamma \Rightarrow (\Sigma \Rightarrow T)$ and $\vdash_{\texttt{SC}_\infty} \Gamma, \Sigma \Rightarrow T$. 
    By the induction hypothesis, we have that $\vdash_{\texttt{G3C}} \Gamma^*, \Sigma^* \Rightarrow T^*$. Now, $\Delta^*\Rightarrow S^*$ is $\Gamma^*\Rightarrow(\Sigma^*\rightarrow T^*)$ which gives $\vdash_{\texttt{G3C}} \Delta^* \Rightarrow S^*$. 
    A similar argument applies to $LI^+$, $RI^-$, and $LI^-$. 

    Now, suppose that the last step is an application of the rule (I). 
    Then $\Delta\Rightarrow S$ is of the form $\Delta \Rightarrow F(A_1,...,A_n)$ and $\vdash_{\texttt{SC}_\infty}\Delta \Rightarrow S_{i 1}$, and ... , and $\vdash_{\texttt{SC}_\infty} \Delta \Rightarrow S_{i s_{i}}$ for some $i\leq t$. 
    By the induction hypothesis, $\vdash_{\texttt{G3C}}\Delta^* \Rightarrow S^*_{i 1}$, and ... , and $\vdash_{\texttt{G3C}} \Delta^* \Rightarrow S^*_{i s_{i}}$. 
    So, $\vdash_{\texttt{G3C}} \Delta^* \Rightarrow S^*_{i 1}\wedge ... \wedge S^*_{i s_{i}}$ and therefore $\vdash_{\texttt{G3C}} \Delta^* \Rightarrow (F(A_1,...,A_n))^*$. 
    Similar arguments apply to (II), (III), and (IV).

\end{proof}

\begin{theorem}
If $\Gamma \cup \{A\}$ is a set of formulas containing only connectives from $\{\wedge, \vee, \rightarrow, {\sim}\} $, then 
    $$\vdash_{\texttt{SC}_\infty} \Delta \Rightarrow A \text{ iff } \vdash_{\texttt{G3C}} \Delta \Rightarrow A .  $$
\end{theorem}

\begin{proof}
    We have $\vdash_{\texttt{G3C}} \Delta \Rightarrow A$ iff $\vdash_{\texttt{SC}_\infty} \Delta^* \Rightarrow A^*$ iff $\vdash_{\texttt{SC}_\infty} \Delta \Rightarrow A.$
\end{proof}

\section{Conclusion}
We used a purely proof-theoretical method of showing functional completeness for the connectives of the connexive logic \texttt{C}. 
This is a desirable approach from a perspective in the spirit of proof-theoretic semantics, which demonstrates that model-theoretic notions are not needed to determine the meaning of logical connectives. 
Furthermore, by considering not only proof but also refutation rule schemata for connectives our approach follows the bilateralist principles of treating proofs and refutations on a par.
The essential difference from similar results for other systems has its source in the connexive understanding of inferences, leading to a different understanding of (meta-)refutation than is usually taken for granted.
With this understanding, though, and according adjustments, we can get the same completeness result for \texttt{C}, i.e., a non-trivial negation inconsistent logic, as we get for other `well-behaved' systems.
Thus, this paper provides a contribution to furthering the understanding and appreciation of the system \texttt{C} and of connexive logics in general by bringing in a bilateralist proof-theoretic semantics perspective.
Also, it shows that having a dependency between definitions of proofs and definitions of refutations is not the key criterion for consistency of a system.
Rather, our background assumptions on how refutation of inferences work are decisive here, namely whether we assume a `classical' or a connexive understanding.

\appendix
\section*{Appendix}

\subsection*{Sequent calculus for \texttt{C}}

The following sequent calculus can be found in \cite[p. 516]{OmoriWansing2020}.
\vspace{0.5cm}

{\Large \textbf{\texttt{G3C}}}

\vspace{0.2cm}

\textbf{Zero-premise rules:}

\vspace{0.2cm}

\quad 
\infer[\scriptstyle Rf]{\Gamma, p \Rightarrow p}{}
\quad \quad
\infer[\scriptstyle Rf^{{\sim}}]{\Gamma, {\sim} p \Rightarrow {\sim} p}{}

\vspace{0.2cm}

\textbf{Rules for $\wedge$:}

\vspace{0.2cm}
\quad   
\infer[\scriptstyle R \wedge]{\Gamma \Rightarrow A \wedge B}{\Gamma \Rightarrow A \quad \quad \Gamma \Rightarrow B}
\quad  \quad
\infer[\scriptstyle L\wedge]{\Gamma, A \wedge B \Rightarrow C}{\Gamma, A, B\Rightarrow C}

\vspace{0.2cm}

\textbf{Rules for $\vee$:}

\vspace{0.2cm}
\quad 
\infer[\scriptstyle R\vee_1]{\Gamma \Rightarrow A \vee B}{\Gamma \Rightarrow A}
\quad \hspace{-0.35cm}
\infer[\scriptstyle R\vee_2]{\Gamma \Rightarrow A \vee B}{\Gamma \Rightarrow B}
\quad  \hspace{-0.35cm}
\infer[\scriptstyle L\vee]{\Gamma, A \vee B \Rightarrow C}{\Gamma, A \Rightarrow C \quad \Gamma, B \Rightarrow C}

\vspace{0.2cm}

\textbf{Rules for $\rightarrow$:}

\vspace{0.2cm}
\quad
\infer[\scriptstyle R\rightarrow]{\Gamma \Rightarrow A \rightarrow B}{\Gamma, A \Rightarrow B}
\quad \quad
\infer[\scriptstyle L\rightarrow]{\Gamma, A \rightarrow B \Rightarrow C}{\Gamma, A \rightarrow B \Rightarrow A \quad \quad \Gamma, B \Rightarrow C}

\vspace{0.2cm}

\textbf{Rules for ${\sim}$:}

\vspace{0.2cm}
\quad  
\infer[\scriptstyle R{\sim} {\sim}]{\Gamma \Rightarrow {\sim}{\sim} A}{\Gamma \Rightarrow A}
\quad \hspace{-0.35cm}
\infer[\scriptstyle L {\sim} {\sim}]{\Gamma, {\sim}{\sim} A \Rightarrow C}{\Gamma, A \Rightarrow C}

\vspace{0.3cm}
\quad  
\infer[\scriptstyle R{\sim}\wedge_{1}]{\Gamma \Rightarrow {\sim}(A \wedge B)}{\Gamma \Rightarrow {\sim} A}
\quad \hspace{-0.35cm}
\infer[\scriptstyle R{\sim}\wedge_{2}]{\Gamma \Rightarrow {\sim}(A \wedge B)}{\Gamma \Rightarrow {\sim} B}
\quad \hspace{-0.35cm}
\infer[\scriptstyle L{\sim}\wedge]{\Gamma, {\sim}(A \wedge B) \Rightarrow C}{\Gamma, {\sim} A \Rightarrow C \quad \Gamma, {\sim} B \Rightarrow C}

\vspace{0.3cm}
\quad  
\infer[\scriptstyle R{\sim}\vee]{\Gamma \Rightarrow {\sim}(A \vee B)}{\Gamma \Rightarrow {\sim} A \quad \quad \Gamma \Rightarrow {\sim} B}
\quad  \quad
\infer[\scriptstyle L{\sim}\vee]{\Gamma, {\sim}(A \vee B) \Rightarrow C}{\Gamma, {\sim} A, {\sim} B \Rightarrow C}

\vspace{0.3cm}
\quad  
\infer[\scriptstyle R{\sim}\rightarrow]{\Gamma \Rightarrow {\sim}(A \rightarrow B)}{\Gamma, A \Rightarrow {\sim} B}
\quad \quad
\infer[\scriptstyle L{\sim}\rightarrow]{\Gamma, {\sim}(A \rightarrow B) \Rightarrow C}{\Gamma, {\sim}(A \rightarrow B)\Rightarrow A \quad \quad \Gamma, {\sim} B \Rightarrow C}

\end{document}